\documentclass[journal, twoside]{IEEEtran} 

\usepackage{cite}

\usepackage{amsmath,amssymb,amsthm}
\usepackage{flushend}

\newtheorem*{theorem}{Theorem}

\newtheorem{proposition}{Proposition}
\newtheorem{lemma}{Lemma}
\newcommand{\E}{\mathbb{E}}
\renewcommand{\P}{\mathbb{P}}
\newcommand{\R}{\mathbb{R}}
\newcommand{\J}{\mathbf{J}} 

\usepackage{url}

\begin{document}

\title{Yet Another Proof of the Entropy Power Inequality} 


\author{Olivier~Rioul,~\IEEEmembership{Member,~IEEE}
\thanks{O. Rioul is with the Department
of Communication and \hbox{Electronics}, LTCI, Telecom ParisTech, Paris-Saclay University, Paris, France. \hbox{E-mail}: {olivier.rioul@telecom-paristech.fr}}
}




\maketitle

\begin{abstract}
Yet another simple proof of the entropy power inequality is given, which avoids both the integration over a path of Gaussian perturbation and the use of Young's inequality with sharp constant or R\'enyi entropies. The proof is based on a simple change of variables,  is formally identical in one and several dimensions, and easily settles the equality case.
\end{abstract}


%

\section{Introduction}
%
%
%
%
\IEEEPARstart{T}{he} entropy power inequality (EPI) was stated by Shannon~\cite{Shannon48} in the form
\begin{equation}\label{epishannon}
e^{\frac{2}{n}h(X+Y)} \geq e^{\frac{2}{n}h(X)}+e^{\frac{2}{n}h(Y)}
\end{equation}
for any independent $n$-dimensional  random vectors $X,Y\in\R^n$ with densities and finite second moments, with equality if and only if $X$ and $Y$ are Gaussian with proportional covariances. Shannon gave an incomplete proof; the first complete proof was given by Stam~\cite{Stam59} using properties of Fisher's information.
A detailed version of Stam's proof was given by Blachman~\cite{Blachman65}. A very different proof was provided by Lieb~\cite{Lieb78} using Young's convolutional inequality with sharp constant. Dembo, Cover and Thomas~\cite{DemboCoverThomas91} provided a clear exposition of both Stam's and Lieb's proofs. Carlen and Soffer
gave an interesting variation of Stam's proof for one-dimensional variables~\cite{CarlenSoffer91}. 
Szarek and Voiculescu~\cite{SzarekVoiculescu00} gave a proof related to Lieb's but based on a variant of the Brunn-Minkowski inequality. Guo, Shamai and Verd\'u gave another proof based on the \hbox{I-MMSE} relation~\cite{GuoShamaiVerdu06,VerduGuo06}. A similar proof based on a relation between divergence and causal MMSE was given by Binia~\cite{Binia07}. Yet another proof based on properties of mutual information was proposed in~\cite{Rioul07,Rioul11}. A more involved proof based on a stronger form of the EPI that uses spherically symmetric rearrangements, also related to Young's inequality with sharp constant, was recently given by Wang and Madiman~\cite{WangMadiman14}. 

As first noted by Lieb~\cite{Lieb78}, the above Shannon's formulation~\eqref{epishannon} of the EPI is equivalent to 
\begin{equation}\label{epi}
h(\sqrt{\lambda} X + \sqrt{1-\lambda}\, Y) \geq \lambda h(X) +(1-\lambda) h(Y)
\end{equation}
for any $0<\lambda<1$. 
All available proofs of the EPI used this form\footnote{Stam's original proof~\cite{Stam59} is an exception, but it was later simplified by Dembo, Cover and Thomas~\cite{DemboCoverThomas91} using this form.}.
Proofs of the equivalence can be found in numerous papers, e.g.,~\cite[Thms. 4, 6, 7]{DemboCoverThomas91}, \cite[Lemma~1]{VerduGuo06}, \cite[Prop.~2]{Rioul11}, and~\cite[Thm~2.5]{MadimanMelbourneXu16}.  

There are a few technical difficulties for proving~\eqref{epi} which are not always explicitly stated in previous proofs. First of all, one should check that for any random vector $X$ with finite second moments, the differential entropy $h(X)$ is always well-defined---even though it could be equal to $-\infty$. This is a consequence of~\cite[Prop.~1]{Rioul11}; see also Appendix~\ref{A} for a precise statement and proof.
Now if both independent random vectors $X$ and $Y$ have densities and finite second moments, so has $\sqrt{\lambda} X + \sqrt{1-\lambda}\, Y$ and both sides of~\eqref{epi} are well-defined. Moreover, if either $h(X)$ or $h(Y)$ equals $-\infty$ then~\eqref{epi} is obviously satisfied. Therefore, one can always assume that $X$ and $Y$ have \emph{finite} differential entropies\footnote{A nice discussion of general necessary and sufficient conditions for the EPI~\eqref{epishannon} can be found in~\cite[\S~V-VI]{BobkovChistyakov15}.}.

Another technical difficulty is the requirement for smooth densities. More precisely, as noted in~\cite[Rmk.~10]{WangMadiman14} some previous proofs use implicitly that for any $X$ with arbitrary density and finite second moments and any  Gaussian\footnote{Throughout this paper we assume that Gaussian random vectors are non-degenerate (have non-singular covariance matrices). 
}
$Z$ independent of $X$,
\begin{equation}\label{linder}
\lim_{t\downarrow 0} h(X+\sqrt{t}Z) = h(X). 
\end{equation}
This was proved explicitly in~\cite[Lemma~3]{Rioul11} and~\cite[Thm.~6.2]{WangMadiman14} using the lower-semicontinuity of divergence;
 the same result can also be found in previous works that were not directly related to the EPI~\cite[Eq.~(51)]{BiniaZakaiZiv74},~\cite[Proof of Lemma~1]{Barron86},~\cite[Proof of Thm~1]{LinderZamir94}.

As a consequence, it is sufficient to prove the EPI for random vectors of the form $X+\sqrt{t}Z$ ($t>0$). Indeed, letting $Z'$ be an independent copy of $Z$ such that $(Z,Z')$ is independent of $(X,Y)$,
the EPI written for $X+\sqrt{t}Z$ and $Y+\sqrt{t}Z'$ reads
\begin{multline*}
h(\sqrt{\lambda} X + \sqrt{1-\lambda}\, Y + \sqrt{t} Z'') \\\geq \lambda h(X+\sqrt{t}Z) +(1-\lambda) h(Y+\sqrt{t}Z') 
\end{multline*}
where $Z''=\sqrt{\lambda} Z + \sqrt{1-\lambda}\, Z'$ is again identically distributed as $Z$ and $Z'$. Letting $t\to 0$ we obtain the general EPI~\eqref{epi}\footnote{A similar observation was done in~\cite{CarlenSoffer91} in a different context of the Ornstein-Uhlenbeck semigroup (instead of the heat semigroup).}.
Now, for any random vector $X$ and any $t>0$, $X+\sqrt{t}Z$ has a continuous and positive density. This can be seen using the properties of the characteristic function, similarly as in~\cite[Lemma~1]{Rioul11}; see Appendix~\ref{B} for a precise statement and proof.
Therefore, as already noticed in~\cite[\S~XI]{WangMadiman14}, one can always assume that $X$ and $Y$ have \emph{continuous, positive} densities. 

One is thus led to prove the following version of the EPI.

\begin{theorem}[EPI]
Let $X,Y$ be independent random vectors with continuous, positive densities and finite differential entropies and second moments.
For any $0<\lambda<1$,
\begin{equation}\tag{\ref{epi}}
h(\sqrt{\lambda} X + \sqrt{1-\lambda}\, Y) \geq \lambda h(X) +(1-\lambda) h(Y)
\end{equation}
with equality if and only if $X,Y$ are Gaussian with identical covariances. 
\end{theorem}

Previous proofs of~\eqref{epi} can be classified into two categories:
\begin{itemize}
\item proofs in~\cite{Stam59,Blachman65,CarlenSoffer91,GuoShamaiVerdu06,VerduGuo06,Binia07,Rioul07,Rioul11} rely on the integration over a path of a continuous Gaussian perturbation of some data processing inequality using either Fisher's information, the minimum mean-squared error (MMSE) or mutual information. As explained in~\cite[Eq.~(10)]{Rioul07}, \cite{Rioul11} and~\cite[Eq.~(25)]{MadimanBarron07}, it is interesting to note that in this context, Fisher's information and MMSE are complementary quantities;
\item proofs in~\cite{Lieb78,SzarekVoiculescu00,WangMadiman13,WangMadiman14} are related to Young's inequality with sharp constant or to an equivalent argumentation using spherically symmetric rearrangements, and/or the consideration of convergence of R\'enyi entropies.
\end{itemize}
It should also be noted that not all of the available proofs of~\eqref{epi} settle the \emph{equality} case---that equality in~\eqref{epi} holds \emph{only} for Gaussian random vectors with identical covariances.  Only proofs from the first category using Fisher's information were shown to capture the equality case. This was made explicit by Stam~\cite{Stam59}, Carlen and Soffer~\cite{CarlenSoffer91} 
and for more general fractional EPI's by Madiman and Barron~\cite{MadimanBarron07}.

In this paper, a simple proof of the Theorem is given that avoids both the integration over a path of a continuous Gaussian perturbation and the use of Young's inequality, spherically symmetric rearrangements, or R\'enyi entropies. It is based on a ``Gaussian to not Gaussian'' lemma proposed in~\cite{RioulCosta16} and is formally identical in one dimension ($n=1$) and in several dimensions ($n>1$). It also easily settles the equality case.


\section{From Gaussian to Not Gaussian}

The following ``Gaussian to not Gaussian'' lemma~\cite{RioulCosta16} will be used here only in the case where 
$X^*$ is a $n$-variate Gaussian vector, e.g., $X^*\sim\mathcal{N}(0,\mathbf{I})$, but holds more generally as $X^*$ needs not be Gaussian.

\begin{lemma}\label{NG2G}
Let $X=(X_1,\ldots,X_n)$ and $X^*=(X^*_1,\ldots,X^*_n)$ be any two $n$-dimensional random vectors in $\R^n$ with continuous, positive densities. 
There exists a diffeomorphism $\Phi$ whose Jacobian matrix is triangular with positive diagonal elements such that $X$ has the same distribution as~$\Phi(X^*)$.
\end{lemma}

For completeness we present two proofs in the Appendix. The first proof in Appendix~ \ref{NG2G1} follows Kn\"othe~\cite{Knothe57}. The second proof  in Appendix~ \ref{NG2G2} is based on the (multivariate) inverse sampling method.

The essential content of this lemma is well known  in the theory of convex bodies~\cite[p.~126]{MilmanSchechtman86},%
%
%
~\cite[Thm.~3.4]{GiannopoulosMilman04}, \cite[Thm.~1.3.1]{ArtsteinGiannopoulosMilman15}
 where $\Phi$ is known as  the \emph{Kn\"othe map} between two convex bodies. The difference with Kn\"othe's map is that in Lemma~\ref{NG2G}, the determinant of the Jacobian matrix need not be constant. 
The Kn\"othe map is also closely related to the so-called Kn\"othe-Rosenblatt coupling in optimal transport theory~\cite{Villani03,Villani08}, and there is a large literature of optimal transportation arguments for geometric-functional inequalities such as the Brunn-Minkowki, isoperimetric, sharp Young, sharp Sobolev and Pr\'ekopa-Leindler inequalities.  The Kn\"othe map was used in the original paper by Kn\"othe~\cite{Knothe57} to generalize the Brunn-Minkowski inequality, by Gromov in~\cite[Appendix I]{MilmanSchechtman86} to obtain isoperimetric inequalities on manifolds and by Barthe~\cite{Barthe98} to prove the sharp Young's inequality. In a similar vein, other transport maps such as the Brenier map were used in~\cite{CorderoErausquinNazaretVillani04} for sharp Sobolev and Gagliardo-Nirenberg inequalities and in~\cite{CorderoErausquinMcCannSchmuckenschlager06} for a generalized Pr\'ekopa-Leindler inequality on manifolds with lower Ricci curvature bounds. Since the present paper was submitted, the Brenier map has also been applied to the stability of the EPI for log-concave densities~\cite{CourtadeFathiPananjady16}. 
All the above-mentionned geometric-functional inequalities are known to be closely related to the EPI (see e.g.,~\cite{DemboCoverThomas91}), and it is perhaps not too surprising to expect a direct proof of the EPI using an optimal transportation argument---namely, Kn\"othe map---which is what this paper is about.

Let $\Phi'$ be the \emph{Jacobian} (i.e., the determinant of the Jacobian matrix) of $\Phi$.
Since $\Phi'>0$, the usual change of variable formula reads
\begin{equation}\label{cvf}
\int f(x) \,\mathrm{d}x  =\int f(\Phi(x^*)) \Phi'(x^*) \,\mathrm{d}x^*.
\end{equation}
A simple application of this formula gives the following well-known lemma which was used in~\cite{RioulCosta16}.
\begin{lemma}\label{CVH}
For any diffeomorphism $\Phi$ with positive Jacobian $\Phi'>0$, if $h\bigl(\Phi(X^*)\bigr)$ is finite,
\begin{equation}\label{cvh}
h\bigl(\Phi(X^*)\bigr)= h(X^*) + \E \{ \log \Phi'(X^*) \}  .
\end{equation}
\end{lemma}
The proof is given for completeness.
\begin{proof}
Let $f(x)$ be the density of $\Phi(X^*)$ so that $g(x^*)=f(\Phi(x^*)) \Phi'(x^*)$ is the density of $X^*$. Then we have
$\int f(x) \log f(x)\,\mathrm{d}x  =\int f(\Phi(x^*)) \log  f(\Phi(x^*))  \,\cdot \Phi'(x^*) \mathrm{d}x^* =\int g(x^*) \log \bigl(g(x^*)/\Phi'(x^*)\bigr)\,\mathrm{d}x^*$ which yields~\eqref{cvh}.
\end{proof}

\section{Proof of the Entropy Power Inequality}

Let $X^*,Y^*$ be \emph{any} i.i.d. Gaussian random vectors, 
e.g., $\sim\mathcal{N}(0,\mathbf{I})$. For any $0<\lambda<1$, $\sqrt{\lambda} X^* + \sqrt{1-\lambda}\, Y^*$ 
is identically distributed
as $X^*$ and $Y^*$ and, therefore,
\begin{equation}
h(\sqrt{\lambda} X^* + \sqrt{1-\lambda}\, Y^*) = \lambda h(X^*) +(1-\lambda) h(Y^*).
\end{equation}
Subtracting both sides from both sides of~\eqref{epi} one is led to prove that 
\begin{multline}\label{tobeproved}
h(\sqrt{\lambda} X + \sqrt{1-\lambda}\, Y) - h(\sqrt{\lambda} X^* + \sqrt{1-\lambda}\, Y^*)   \\
\geq  \lambda \bigl(h(X)-h(X^*)\bigr) +(1-\lambda) \bigl(h(Y)-h(Y^*)\bigr).
\end{multline}
Let $\Phi$ be as in Lemma~\ref{NG2G}, so that $X$ has the same distribution as $\Phi(X^*)$. Similarly let $\Psi$ be such that $Y$ has the same distribution as $\Psi(Y^*)$.
Since $\sqrt{\lambda} X + \sqrt{1-\lambda}\, Y$ is identically distributed as $\sqrt{\lambda} \Phi(X^*) + \sqrt{1-\lambda}\, \Psi(Y^*)$,
\begin{align}\label{1}
&h(\sqrt{\lambda} X + \sqrt{1-\lambda}\, Y) - h(\sqrt{\lambda} X^* + \sqrt{1-\lambda}\, Y^*) 
\notag\\&= h\bigl(\sqrt{\lambda} \Phi(X^*) + \sqrt{1-\lambda}\, \Psi(Y^*)\bigr) - h(\sqrt{\lambda} X^* + \sqrt{1-\lambda}\, Y^*) .
\end{align}
On the other hand,  by Lemma~\ref{CVH},
\begin{align}
&\lambda \bigl(h(X)-h(X^*)\bigr) +(1-\lambda) \bigl(h(Y)-h(Y^*)\bigr) \notag\\
&=\lambda \bigl(h\bigl(\Phi(X^*)\bigr)-h(X^*)\bigr) +(1-\lambda) \bigl(h\bigl(\Psi(Y^*)\bigr)-h(Y^*)\bigr)
\notag\\&= \E \{  \lambda\log \Phi'(X^*) +(1-\lambda) \log \Psi'(Y^*) \}.\label{2}
\end{align}
Thus both sides of~\eqref{tobeproved} have been rewritten in terms of the Gaussian $X^*$ and $Y^*$.
We now compare~\eqref{1} and~\eqref{2}. Toward this aim we make the change of variable $(X^*,Y^*)\to (\tilde{X},\tilde{Y})$ where
\begin{equation}
\begin{cases}
\tilde{X} =  \sqrt{\lambda} X^* + \sqrt{1-\lambda}\, Y^* \\
\tilde{Y} =  -\sqrt{1-\lambda}\, X^* + \sqrt{\lambda} Y^* .
\end{cases} 
\end{equation}
Again $\tilde{X},\tilde{Y}$ are i.i.d. Gaussian 
and
\begin{equation}
\begin{cases}
X^* =  \sqrt{\lambda} \tilde{X} - \sqrt{1-\lambda}\, \tilde{Y} \\
Y^* =  \sqrt{1-\lambda}\, \tilde{X} + \sqrt{\lambda} \tilde{Y} .
\end{cases} 
\end{equation}
To simplify the notation define
\begin{equation}
\begin{split}
\Theta_{\tilde{y}}(\tilde{x}) &=  \sqrt{\lambda} \Phi(\sqrt{\lambda} \tilde{x} - \sqrt{1-\lambda}\, \tilde{y}) 
\\&\qquad+ \sqrt{1-\lambda}\, \Psi(\sqrt{1-\lambda}\, \tilde{x} + \sqrt{\lambda} \tilde{y}).
\end{split}
\end{equation}
Then~\eqref{1} becomes
\begin{equation}
\begin{split}
h(\sqrt{\lambda} X + \sqrt{1-\lambda}\, Y) - h(\sqrt{\lambda} X^* + \sqrt{1-\lambda}\, Y^*) \\
= h\bigl( \Theta_{\tilde{Y}}(\tilde{X})\bigr) - h(\tilde{X}). \label{3}
\end{split}
\end{equation}
Here Lemma~\ref{CVH} cannot be applied directly because $\Theta_{\tilde{Y}}(\tilde{X})$ is not a deterministic function of $\tilde{X}$. But since conditioning reduces entropy,
\begin{equation}
h\bigl( \Theta_{\tilde{Y}}(\tilde{X})\bigr) \geq h\bigl( \Theta_{\tilde{Y}}(\tilde{X})\big| \tilde{Y}\bigr) \label{CRE}
\end{equation}
Now for fixed $\tilde{y}$, since  $\Phi$ and $\Psi$ have triangular Jacobian matrices with positive diagonal elements, the Jacobian matrix of $\Theta_{\tilde{y}}$ is also triangular with positive diagonal elements. Thus, by Lemma~\ref{CVH}, 
\begin{equation}
h\bigl( \Theta_{\tilde{Y}}(\tilde{X})\big| \tilde{Y}=\tilde{y}\bigr) -h(\tilde{X}) =  \E \{ \log \Theta_{\tilde{y}}'(\tilde{X}) \}
\end{equation}
where $\Theta'_{\tilde{y}}$ is the Jacobian of the transformation $\Theta_{\tilde{y}}$.
Since $\tilde{X}$ and $\tilde{Y}$ are independent, averaging over $\tilde{Y}$ yields
\begin{equation}
h\bigl( \Theta_{\tilde{Y}}(\tilde{X})\big| \tilde{Y}\bigr) -h(\tilde{X}) =  \E \{ \log \Theta_{\tilde{Y}}'(\tilde{X}) \}.
\end{equation}
Therefore, by~\eqref{3}-\eqref{CRE}
\begin{equation}
\begin{split}
h(\sqrt{\lambda} X + \sqrt{1-\lambda}\, Y) - h(\sqrt{\lambda} X^* + \sqrt{1-\lambda}\, Y^*) \\
\geq  \E \{ \log \Theta_{\tilde{Y}}'(\tilde{X}) \}.
\end{split}\label{4}
\end{equation}
On the other hand,~\eqref{2} becomes
\begin{align}
\lambda \bigl(h&(X)-h(X^*)\bigr) +(1-\lambda) \bigl(h(Y)-h(Y^*)\bigr)  \notag\\
\begin{split}
&=\E \{  \lambda\log \Phi'(\sqrt{\lambda} \tilde{X} - \sqrt{1-\lambda}\, \tilde{Y}) \\&\qquad+(1-\lambda) \log \Psi'(\sqrt{1-\lambda}\, \tilde{X} + \sqrt{\lambda} \tilde{Y}) \} 
\end{split}
\\
\begin{split}
&= \sum_{i=1}^n \E \bigl\{ \lambda \log \frac{\partial \Phi_i}{\partial x_i}(\sqrt{\lambda} \tilde{X} - \sqrt{1-\lambda}\, \tilde{Y}) \\&\qquad\qquad+ (1-\lambda) \log \frac{\partial \Psi_i}{\partial y_i}(\sqrt{1-\lambda}\, \tilde{X} + \sqrt{\lambda} \tilde{Y}) \bigr\} 
\end{split}
\\   \label{jensen}
\begin{split}
&\leq  \sum_{i=1}^n \E \log \bigl\{ \lambda  \frac{\partial \Phi_i}{\partial x_i}(\sqrt{\lambda} \tilde{X} - \sqrt{1-\lambda}\, \tilde{Y}) \\&\qquad\qquad+ (1-\lambda)  \frac{\partial \Psi_i}{\partial y_i}(\sqrt{1-\lambda}\, \tilde{X} + \sqrt{\lambda} \tilde{Y}) \bigr\} 
\end{split}
\\
&= \sum_{i=1}^n \E \log    \frac{ \partial\bigl(\Theta_{\tilde{Y}}\bigr)_{i}}{\partial \tilde{x}_i}(\tilde{X})
= \E \log \Theta'_{\tilde{Y}}(\tilde{X}) \label{triangle}\\
&\leq h(\sqrt{\lambda} X + \sqrt{1-\lambda}\, Y) - h(\sqrt{\lambda} X^* + \sqrt{1-\lambda}\, Y^*) \label{final}
\end{align}
where in~\eqref{jensen} we have used Jensen's inequality $\lambda\log a+(1-\lambda)\log b\leq \log(\lambda a+(1-\lambda) b)$ on each component, in~\eqref{triangle} the fact that  the Jacobian matrix of $\Theta_{\tilde{y}}$ is triangular with positive diagonal elements, and~ \eqref{final} is~\eqref{4}.
This proves~\eqref{epi}.

\section{The Case of Equality}

Equality in~\eqref{epi} holds if and only if both~\eqref{CRE} and~\eqref{jensen} are equalities. Equality in~\eqref{jensen} holds if and only if for all $i=1,2\ldots,n$,
\begin{equation}
 \frac{\partial \Phi_i}{\partial x_i}(X^*) = \frac{\partial \Psi_i}{\partial y_i}(Y^*) \quad\text{a.e.}
\end{equation}
Since $X^*$ and $Y^*$ are independent Gaussian random vectors this implies that $\dfrac{\partial \Phi_i}{\partial x_i}$ and $\dfrac{\partial \Psi_i}{\partial y_i}$ are constant and equal. Thus in particular $\Theta'_{\tilde{y}}$ is constant.
Now equality in~\eqref{CRE} holds if and only if $\Theta_{\tilde{Y}}(\tilde{X})$ is independent of $\tilde{Y}$, thus $\Theta_{\tilde{y}}(\tilde{X})=\Theta(\tilde{X})$ does not depend on the particular value of $\tilde{y}$.
Thus for all $i,j=1,2,\ldots,n$,
\begin{equation}
\begin{split}
0&=\frac{\partial (\Theta_{\tilde{y}}(\tilde{X}))_i}{\partial \tilde{y}_j}\\& = -\sqrt{\lambda}\sqrt{1-\lambda} \,
\frac{\partial \Phi_i}{\partial x_j}(\sqrt{\lambda} \tilde{X} - \sqrt{1-\lambda}\, \tilde{Y})
\\&\quad\,+\sqrt{1-\lambda}\sqrt{\lambda} \frac{\partial \Psi_i}{\partial y_j}(\sqrt{1-\lambda}\, \tilde{X} + \sqrt{\lambda} \tilde{Y})
\end{split}
\end{equation}
which implies 
\begin{equation}
 \frac{\partial \Phi_i}{\partial x_j}(X^*) = \frac{\partial \Psi_i}{\partial y_j}(Y^*) \quad\text{a.e.},
\end{equation}
hence $\dfrac{\partial \Phi_i}{\partial x_j}$ and $\dfrac{\partial \Psi_i}{\partial y_j}$ are constant and equal for any $i,j=1,2,\ldots,n$.
Therefore, $\Phi$ and $\Psi$ are linear transformations, equal up to an additive constant. It follows that $\Phi(X^*)$ and $\Phi(Y^*)$ (hence $X$ and $Y$) are Gaussian with identical covariances.
This ends the proof of the Theorem. \hfill\qedsymbol

\bigskip

Extensions of similar ideas when $X^*,Y^*$ need not be Gaussian can be found in~\cite{Rioul17a}.

\appendices
\section{}\label{A}

The differential entropy $h(X)=-\int f\log f$ of a random vector $X$ with density $f$ is not always well-defined because the negative and positive parts of the integral might be both infinite, as in the example $f(x)=1/(2x\log^2 x)$ for $0<x<1/\mathrm{e}$ and $\mathrm{e}<x<+\infty$, and $=0$ otherwise~\cite{Rioul11}. 

\begin{proposition}
Let $X$ be an random vector with density $f$ and finite second moments. Then $h(X)=-\int f \log f$ is well-defined and $<+\infty$.
\end{proposition}

\begin{proof}
Let $Z$ be any Gaussian vector with density $g>0$. On one hand, since $X$ has finite second moments, the integral $\int f \log g$ is finite. 
On the other hand, since $g$ never vanishes, the probability measure of $X$ is absolutely continuous with respect to that of $Z$. Therefore, the divergence $D(f\|g)$ is equal to the integral $\int f \log (f/g)$. Since the divergence is non-negative, it follows that $-\int f\log f =  -\int f \log g - D(f\|g)\leq  -\int f \log g$ is well-defined and $<+\infty$ (the positive part of the integral is finite).
\end{proof}

\section{}\label{B}

It is stated in~\cite[Appendix~II~ A]{GengNair14}  that strong smoothness properties of distributions of $Y=X+Z$ for independent Gaussian $Z$ are ``very well known in certain mathematical circles'' but it seems difficult to find a reference.

The following result is stated for an arbitrary random vector $X$. It is not required that $X$ have a density. It could instead follow e.g., a discrete distribution. 

\begin{proposition}
Let $X$ be any random vector and $Z$ be any independent Gaussian vector with density $g>0$. Then $Y=X+Z$ has a bounded, positive, indefinitely differentiable (hence continuous) density that tends to zero at infinity, whose all derivatives are also bounded and tend to zero at infinity.
\end{proposition}

\begin{proof}
Taking characteristic functions, $\E(e^{it\cdot Y}) = \E(e^{it\cdot X}) \cdot \E(e^{it\cdot Z})$, where $\hat{g}(t)=\E(e^{it\cdot Z})$ is the Fourier transform of the Gaussian density $g$. Now $\hat{g}(t)$ is also a Gaussian function with exponential decay at infinity and
$|\E(e^{it\cdot Y}) |\leq  |\E(e^{it\cdot X})| \cdot |\E(e^{it\cdot Z})| \leq \E(|e^{it\cdot X}|)| \cdot |\hat{g}(t)| =|\hat{g}(t)|$. Therefore, the Fourier transform of the probability measure of $Y$ (which is always continuous) also has exponential decay at infinity. 
%
%
In particular, this Fourier transform is integrable, and by the Riemann-Lebesgue lemma, $Y$ has a bounded continuous density which tends to zero at infinity. Similarly, for any monomial\footnote{Here we use the multi-index notation $t^\alpha=t_1^{\alpha_1}t_1^{\alpha_1}\cdots t_n^{\alpha_n}$.} $t^\alpha$, $(it)^\alpha\E(e^{it\cdot Y})$ is integrable and is the Fourier transform of the  $\alpha$th partial derivative of the density of $Y$, which is, therefore, also bounded continuous and tends to zero at infinity.

It remains to prove that the density of $Y$ is positive. Let $Z_1$, $Z_2$ be independent Gaussian random vectors with density~$\phi$ equal to that of $Z/\sqrt{2}$ so that $Z$ has the same distribution as $Z_1+Z_2$. By what has just been proved, $X+Z_1$ follows a continuous density $f$. Since $Y$ has the same distribution as $(X+Z_1) + Z_2$, its density is equal to the convolution product $f\!*\!\phi\,(y) = \int_\R \phi(z) f(y-z)\,\mathrm{d}z$. Now $\phi$ is positive,  and for any $y\in\R^n$,  $\int_\R \phi(z) f(y-z)\,\mathrm{d}z=0$  would imply that $f$ vanishes identically, which is impossible.
\end{proof}

\section{First Proof of Lemma~\ref{NG2G}}\label{NG2G1}

We use the notation $f$ for densities (p.d.f.'s).
In the first dimension,
for each $x^*_1\in\R$, define $\Phi_1(x^*_1)$ such that 
\begin{equation}
\int_{-\infty}^{\Phi_1(x^*_1)} f_{X_1} = \int_{-\infty}^{x^*_1} f_{X^*_1}.
 \end{equation}
Since the densities are continuous and positive, $\Phi_1$ is continuously differentiable and increasing; differentiating gives
\begin{equation}
f_{X_1}(\Phi_1(x^*_1)) \;  \frac{\partial \Phi_1}{\partial x^*_1}(x^*_1) = f_{X^*_1}(x^*_1)
\end{equation}
which proves the result in one dimension: $X_1$ has the same distribution as $\Phi_1(X^*_1)$ where $\dfrac{\partial \Phi_1}{\partial x^*_1}$ is positive.

In the first two dimensions, for each $x^*_1,x^*_2$ in $\R$, define  $\Phi_2(x^*_1,x^*_2)$ such that 
\begin{equation}
\int_{-\infty}^{\Phi_2(x^*_1,x^*_2)}\!\!\! f_{X_1,X_2}(\Phi_1(x^*_1),\,\cdot\,) \;  \frac{\partial \Phi_1}{\partial x^*_1}(x^*_1) = 
\int_{-\infty}^{x^*_2}\!\! f_{X^*_1,X^*_2}(x^*_1,\,\cdot\,).
\end{equation} 
Again $\Phi_2$ is continuously differentiable and increasing in $x^*_2$; differentiating gives
\begin{multline}
f_{X_1,X_2}(\Phi_1(x^*_1),\Phi_2(x^*_1,x^*_2)) \;  \frac{\partial \Phi_1}{\partial x^*_1}(x^*_1)  \frac{\partial \Phi_2}{\partial x^*_2}(x^*_1,x^*_2) \\=  f_{X^*_1,X^*_2}(x^*_1,x^*_2)
\end{multline} 
which proves the result in two dimensions. Continuing in this manner we arrive at
\begin{multline}
\begin{split}
&f_{X_1,X_2,\ldots,X_n}(\Phi_1(x^*_1),\Phi_2(x^*_1,x^*_2),\ldots,\Phi_n(x^*_1,x^*_2,\ldots,x^*_n))
 \\&\qquad\times   \frac{\partial \Phi_1}{\partial x^*_1}(x^*_1)  \frac{\partial \Phi_2}{\partial x^*_2}(x^*_1,x^*_2) \cdots \frac{\partial \Phi_n}{\partial x^*_n}(x^*_1,x^*_2,\ldots,x^*_n) 
 \end{split}
\\=  f_{X^*_1,X^*_2,\ldots,X^*_n}(x^*_1,x^*_2,\ldots,x^*_n) 
\end{multline}
which shows that $X=(X_1,X_2,\ldots,X_n)$ has the same distribution as $\Phi(X^*_1,X^*_2,\ldots,X^*_n)=\bigl(\Phi_1(X^*_1),\Phi_2(X^*_1,X^*_2),\ldots,
\linebreak[1]
\Phi_n(X^*_1,X^*_2,\ldots,X^*_n)\bigr)$. The Jacobian matrix of $\Phi$ has the form
\begin{equation}
\J_\Phi(x_1^*,x^*_2,\ldots,x^*_n)=
\begin{pmatrix}
\frac{\partial \Phi_1}{\partial x^*_1} & 0 & \cdots & 0\\
\frac{\partial \Phi_2}{\partial x^*_1} & \frac{\partial \Phi_2}{\partial x^*_2} & \cdots & 0\\
\hdotsfor{4}\\
\frac{\partial \Phi_n}{\partial x^*_1} & \frac{\partial \Phi_n}{\partial x^*_2} & \cdots & \frac{\partial \Phi_n}{\partial x^*_n}\\
\end{pmatrix}
 \end{equation}
where all diagonal elements are positive since by construction each $\Phi_k$ is increasing in $x^*_k$.\hfill\qedsymbol

\section{Second Proof of Lemma~\ref{NG2G}}\label{NG2G2}

We use the notation $F$ for distribution functions (c.d.f.'s). We also note
$F_{X_2|X_1}(x_2|x_1)=\P(X_2\!\leq\! x_2 \,|\, X_1\!=\!x_1)$ and let $F^{-1}_{X_2|X_1}(\cdot|x_1)$ be the corresponding inverse function in the  argument~$x_2$ for a fixed value of $x_1$. 
Such inverse functions are well-defined since it is assumed that $X$ is a random vector with continuous, positive density.

The inverse transform sampling method is well known for univariate random variables but its multivariate generalization is not. 

\begin{lemma}[Multivariate Inverse Transform Sampling Method (see, e.g., {\cite[Algorithm~2]{CasterEkenberg12}})]
\label{MITSM}
Let $U=(U_1,U_2,\ldots,U_n)$ be uniformly distributed on $[0,1]^n$. The vector $\Phi(U)$ with components
\begin{align}
\begin{split}
\Phi_1(U_1)&=F^{-1}_{X_1}(U_1)\\
\Phi_2(U_1,U_2)&=F^{-1}_{X_2|X_1}(U_2|\Phi_1(U_1))
\end{split}\notag\\
\begin{split}
&\;\;\vdots
\end{split}
\\
\Phi_n(U_1,U_2,\ldots,U_n)&= \notag\\\notag
&  {\mkern-108mu}      F^{-1}_{X_n|X_1,\ldots,X_{n-1}}(U_n|\Phi_1(U_1),\ldots,\Phi_{n-1}(U_1,\ldots,U_{n-1}))
\end{align}
has the same distribution as $X$.
\end{lemma}

\begin{proof}
By inverting $\Phi$, it is easily seen that an equivalent statement is that  the random vector
$\bigl(F_{X_1}(X_1), 
\linebreak[1]
F_{X_2|X_1}(X_2|X_1), \ldots, 
F_{X_n|X_1,\ldots,X_{n-1}}(X_n|X_1,\ldots,X_{n-1})\bigr)$ is uniformly distributed in $[0,1]^n$. Clearly $F_{X_1}(X_1)$ is uniformly distributed in $[0,1]$,
since 
\begin{equation}
\begin{split}
\P(F_{X_1}(X_1)\leq u_1) &= \P(X_1\leq F^{-1}_{X_1}(u_1))  \\&= F_{X_1}\circ F^{-1}_{X_1}(u_1)\\&=u_1.
\end{split}
\end{equation}
Similarly, for any $k>0$ and fixed $x_1,\ldots,x_{k-1}$, $F_{X_k|X_1,\ldots,X_{k-1}}(X_k|{X_1=x_1},{X_2=x_2},\ldots,X_{k-1}=x_{k-1})$ is also uniformly distributed in $[0,1]$. The result follows by the chain rule.
\end{proof}

\begin{proof}[Proof of Lemma~\ref{NG2G}]
By Lemma~\ref{MITSM}, $X$ has the same distribution as $\Phi(U)$, 
where each $\Phi_k(u_1,u_2,\ldots,u_k)$ is increasing in $u_k$.
Similarly  $X^*$ has the same distribution as $\Psi(U)$, where both $\Phi$ and $\Psi$ have (lower) triangular Jacobian matrices $\J_\Phi,\J_\Psi$ with positive diagonal elements. Then $X$ has the same distribution as $\Phi(\Psi^{-1}(X^*))$. By the chain rule for differentiation, the transformation $\Phi\circ\Psi^{-1}$ has Jacobian matrix $(\J_\Phi\circ\Psi^{-1}) \cdot \J_{\Psi^{-1}}=(\J_\Phi\circ\Psi^{-1}) \cdot (\J_\Psi\circ \Psi^{-1})^{-1}$. This product of (lower) triangular matrices with positive diagonal elements and is again (lower) triangular with positive diagonal elements. 
\end{proof}

\section*{Acknowledgment} 
The author would like to thank Max Costa, Tom Courtade and C\'edric Villani for their discussions, Tam\'as Linder for pointing out references~\cite{BiniaZakaiZiv74,Barron86,LinderZamir94} in connection with~\eqref{linder} and the anonymous reviewers for pointing out references~\cite{Barthe98,CorderoErausquinNazaretVillani04,CorderoErausquinMcCannSchmuckenschlager06} related to optimal transport theory.

\vspace*{0ex}

\end{document}